\documentclass[11pt]{article}

\usepackage{microtype}
\usepackage{graphicx}
\usepackage{caption}
\usepackage{subcaption}
\usepackage{theorem}
\usepackage{epsfig}
\usepackage{amssymb}
\usepackage{fullpage}
\usepackage{amsmath}
\usepackage{ifthen}
\usepackage{verbatim}
\usepackage{mdframed}
\usepackage{xspace}
\usepackage{enumitem}
\usepackage[multiple]{footmisc}
\usepackage{url}
\usepackage[colorlinks,linkcolor=blue,citecolor=blue,urlcolor=blue]{hyperref}

\usepackage{tikz}
\usetikzlibrary{shapes,shapes.geometric,arrows,fit,calc,positioning,automata,arrows.meta}
\tikzset{>={Latex[width=2mm,length=2mm]}}
\usepackage{float}

\usepackage[noend,noline]{algorithm2e}
\SetEndCharOfAlgoLine{}
\SetArgSty{}
\SetKwBlock{Repeat}{repeat}{}
\newtheorem{theorem}	 			{Theorem}[section]
\newtheorem{lemma}		[theorem]	{Lemma}

\newtheorem{definition}	 			{Definition}[section]
{\theorembodyfont{\rmfamily} \newtheorem{remark}		[theorem]
{Remark}}
{\theorembodyfont{\rmfamily} \newtheorem{proposition}		[theorem]
{Proposition}}
{\theorembodyfont{\rmfamily} }
{\theorembodyfont{\rmfamily} }
{\theorembodyfont{\rmfamily} }
{\theorembodyfont{\rmfamily} \newtheorem{q}			[theorem]
{Question}}
{\theorembodyfont{\rmfamily} }
{\theorembodyfont{\rmfamily} }
\theoremstyle{break}
{\theorembodyfont{\rmfamily} }

\newenvironment{proof}{\noindent {\em {Proof:}}}{$\blacksquare$\vskip
\belowdisplayskip}

\usepackage{tablefootnote} 
\makeatletter 
\AfterEndEnvironment{mdframed}{%
 \tfn@tablefootnoteprintout%
 \gdef\tfn@fnt{0}%
}
\makeatother

\newlist{exlist}{enumerate}{1}
\setlist[exlist]{label=(\alph*)}

\newcommand{\prob}[2][]{\text{\bf Pr}\ifthenelse{\not\equal{}{#1}}{_{#1}}{}\!\left[#2\right]}
\newcommand{\expect}[2][]{\text{\bf E}\ifthenelse{\not\equal{}{#1}}{_{#1}}{}\!\left[#2\right]}
\newcommand{\var}[2][]{\text{Var}\ifthenelse{\not\equal{}{#1}}{_{#1}}{}\!\left[#2\right]}
\newcommand{\dev}[2][]{\text{StdDev}\ifthenelse{\not\equal{}{#1}}{_{#1}}{}\!\left[#2\right]}

\mdfsetup{frametitlealignment=\centering}
\mdfdefinestyle{offset}{backgroundcolor=white,linecolor=black,innerrightmargin=15pt,leftmargin=23pt,rightmargin=23pt,innertopmargin=.5\baselineskip,innerbottommargin=.5\baselineskip}

\def\sm{\setminus}
\def\sse{\subseteq}

\def\RR{\mathbb{R}}

\def\I{\mathcal{I}}

\def\var{\mbox{Var}}

\def\σ{\mathbf{\sigma}}
\def\β{\mathbf{\beta}}

\newcommand{\MaxCut }{{\sc MaxCut }}

\title{Stability and Recovery for Independence Systems}
\author{Vaggos Chatziafratis\thanks{Computer Science Department, Stanford University}\\{\tt vaggos@stanford.edu}\and Tim Roughgarden\footnotemark[1]\\{\tt tim@cs.stanford.edu}\and Jan Vondrak\thanks{Department of Mathematics, Stanford University}\\{\tt jvondrak@stanford.edu}}
\date{\today}

\begin{document}

\maketitle

\begin{abstract}
Two genres of
  heuristics that are frequently reported to perform much better on
  ``real-world'' instances than in the worst case are {\em greedy
    algorithms} and {\em local search algorithms}.  In this paper, we
  systematically study these two types of algorithms for the problem
  of maximizing a monotone submodular set function subject to
  downward-closed feasibility constraints.  We consider {\em
    perturbation-stable} instances, in the sense of Bilu and
  Linial~\cite{bilu2012stable}, and precisely identify the stability
  threshold beyond which these algorithms are guaranteed to recover
  the optimal solution.  Byproducts of our work include the first
  definition of perturbation-stability for non-additive objective
  functions, and a resolution of the worst-case approximation
  guarantee of local search in $p$-extendible systems.

\end{abstract}
\clearpage
\section{Introduction}\label{sec:introduction}

Designing polynomial-time approximation algorithms with worst-case
guarantees is one of the most common approaches to coping with
$NP$-hard optimization problems.  For many problems, even the
best-achievable worst-case guarantee (assuming $P \neq NP$) is too
weak to be immediately meaningful.
Fortunately, it has been widely observed that most approximation
algorithms typically compute solutions that are much better than their
worst-case approximation guarantee would suggest
(e.g.~\cite{cormode2010set,rego2011traveling}).
Is there a mathematical explanation for this phenomenon?


One line of work addresses this question by restricting attention to
instances that satisfy a {\em stability} condition, stating that there
should be a ``sufficiently prominent'' optimal solution.  
Such conditions are analogs of the ``large margin'' assumptions that are
often made in machine learning theory.
Such assumptions reflect the belief that the instances arising in
practice are ones that have a ``meaningful solution''.  For example,
if we run a clustering algorithm on a data set, it's because we're
expecting that a ``meaningful clustering'' exists.  
The hope is that
formalizing the assumption of a ``meaningful solution'' imposes
additional structure on an instance that provably makes the problem
easier than on worst-case instances.

Several such stability notions have been studied.
In this work, we focus on the most well-studied one, that of {\em
  perturbation-stability}
introduced by Bilu and Linial~\cite{bilu2012stable}.
The idea behind the definition is that the optimal solution should be
robust to small changes in the input (e.g., the edge weights of a
graph).  For if this is not true, then a minor misspecification of the
data (which is often noisy in practice, anyways) can change the output of
the algorithm.  
In data analysis, one is certainly
hoping that the conclusions reached are not sensitive to small errors
in the data.
%
An informal definition of $\gamma$-perturbation-stability (henceforth
simply {\em $\gamma$-stability}) is the following:
\begin{definition} [$\gamma$-stability] Given a weighted graph and an
  optimal solution $S^*$ for some problem, we say that the instance is
  {\em $\gamma$-stable} if $S^*$ remains the unique optimal solution, even
  when each edge weight is increased by an (edge-dependent) factor
  between~1 and~$\gamma$.
\end{definition}
Thus $1$-stability is equivalent to the assumption that the
optimal solution is unique.  The bigger the~$\gamma$, the stronger the
assumption (since there are fewer instances we are required to solve), and hence, the easier the problem.
The basic question is then \textbf{whether sufficiently stable
  instances of computationally hard problems are easier to
  solve}. 
The ultimate goal is to determine the {\em stability threshold} of a
problem:
the smallest value of $\gamma$ such that the problem is
  polynomial-time solvable on $\gamma$-stable instances.
We note that there is no general connection between hardness of
approximation thresholds and stability thresholds of a
problem---depending on  the
problem, each could be larger than the other (e.g.,~\cite{balcan2015k}, where even though asymmetric $k$-center cannot be approximated to any
constant factor, it can be solved optimally under 2-stability).
%
Thus a good approximation algorithm need not recover an optimal
solution in stable instances, and conversely.\footnote{%
For a silly example, consider an algorithm that
  checks if an instance is stable (by brute-force), if so returns the
  optimal solution (computed by brute force), and if not returns a
  terrible solution.  Similarly, consider an $\alpha$-approximation
  algorithm that uses brute force to always output a suboptimal
  solution, in every instance where one   within $\alpha$ of optimal
  exists.  For more natural (and polynomial-time) examples,
  see~\cite{balcan2015k,makarychevSODA14}.}


\subsection{Our Results}

Two genres of algorithms that are frequently reported to perform much
better on ``real-world'' instances than in the worst case are {\em
  greedy algorithms} and {\em local search algorithms}.  The goal of
this paper is to systematically study these two types of algorithms
through the lens of perturbation-stability.  We carry this out for the
rich and well-motivated class of problems that concern maximizing a
monotone submodular set function subject to downward-closed
feasibility constraints (as in e.g.~\cite{lovasz1983submodular,
  nemhauser1978analysis, fisher1978analysis}).  Both greedy and local
search algorithms can be naturally defined for all problems in this
class.  Special cases include~\cite{mestre2006greedy} $k$-dimensional matching, asymmetric
traveling salesman, influence maximization~\cite{kempe2003maximizing},
welfare maximization in combinatorial auctions (with submodular
valuations)~\cite{lehmann2001combinatorial,vondrak2008optimal}, and so
on.

We organize our results along two different axes: whether the
objective function is additive or submodular, and according to the
``complexity'' of the feasibility constraints.  
For the latter, we use the classic notions of the intersection of
$p$ matroids (for a parameter $p$), the more general notion of
$p$-extendible systems (where a feasible solution can accommodate a
new element after deleting at most $p$ old ones), and the still
more general notion of $p$-systems (where the cardinality of maximal
independent sets can only differ by a $p$ factor).
\hyperref[table-results]{Figure \ref{table-results}} summarizes our main results.  
We also prove that all of our results are tight.

\begin{figure}[h!]
	\centering
	\includegraphics[width=14cm,height=6cm]{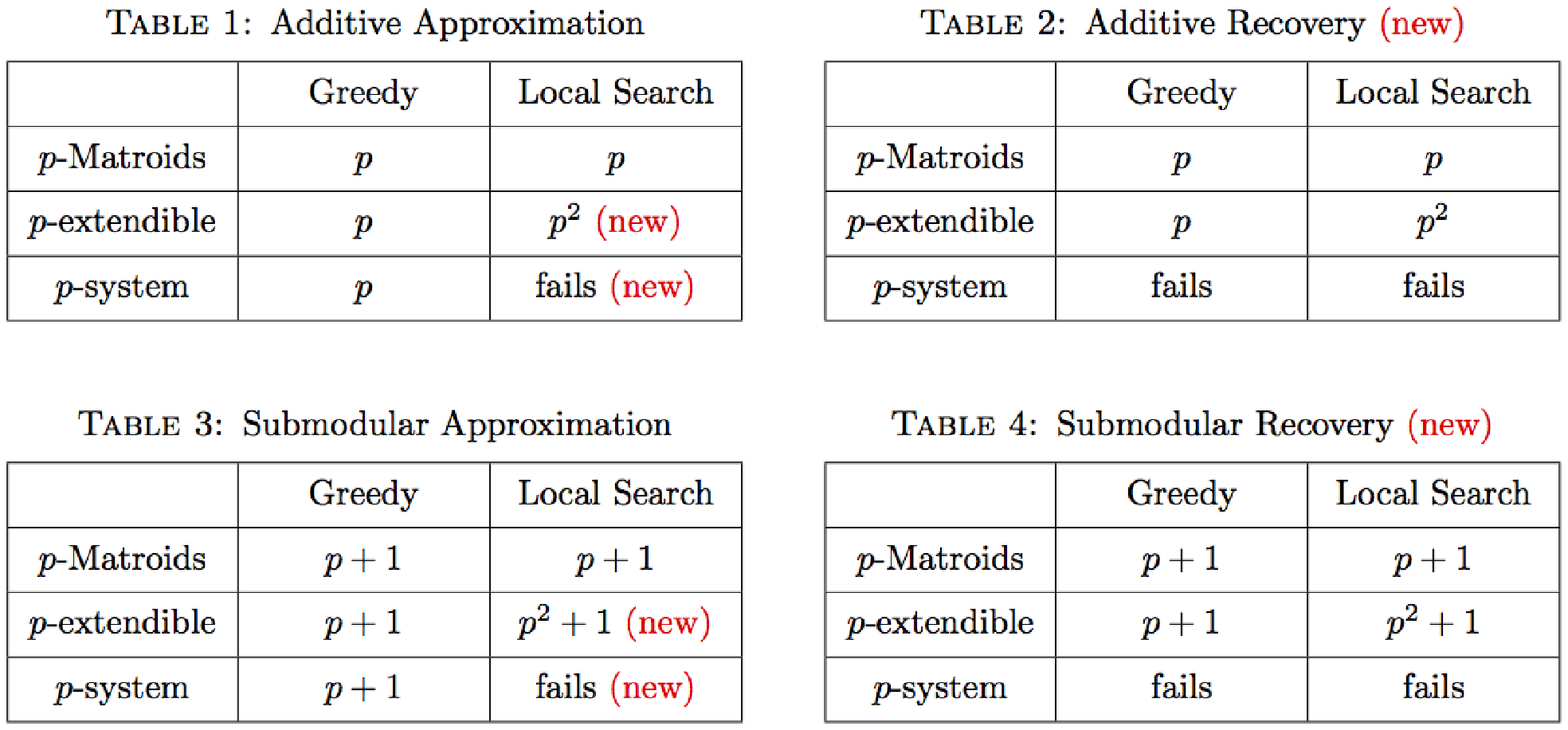}
        \caption{Summary of old and new results. 
On the left we have
          previous approximation results about greedy and local
          search          algorithms~\cite{korte1978analysis,nemhauser1978analysis,fisher1978analysis,
            reichel2007evolutionary} and our new local search
          approximation guarantees. (Each table entry indicates the
          worst-case approximation factor.)  On the right are our recovery
          results for greedy and local search algorithms, with each
          table entry indicating the smallest $\gamma$ such that the
          algorithm is optimal in every $\gamma$-stable instance.
All of the results are tight.}
	\label{table-results}
\end{figure}

\hyperref[sec:additive]{Section \ref{sec:additive}} proves our results
for the greedy algorithm in the case of additive objective functions.
An interesting finding here is that for the most general set systems that
we consider ($p$-systems), the greedy algorithm can have an
infinite stability threshold, even though it is a good worst-case
approximation algorithm. In fact, this crucial difference between approximation and stability also led us to give a different characterization of the $p$-extendible systems. Another interesting differentiation between stability and approximation 
shows up in the case of a uniform matroid (cardinality constraints).

\hyperref[sec:submodular]{Section \ref{sec:submodular}} considers the greedy algorithm for
maximizing a monotone submodular function.  As all previous works on
perturbation-stability have considered only problems with additive
objective functions, here we need to formulate a notion of
perturbation-stability for submodular functions, which boils down to
defining the class of allowable perturbations of a submodular
function~$f$.  The ``sweet spot''---neither too restrictive nor too
permissive---turns out to be the set of perturbed
functions~$\tilde{f}$ such that: (i) $\tilde{f}$ is also monotone and
submodular; (ii) $\tilde{f}$ is a pointwise approximation of $f$
($\tilde{f}(S) \in [f(S),\gamma \cdot f(S)]$ for every $S$); and (iii)
the marginal value of an element~$j$ with respect to a set $S$ can
only go up (in $\tilde{f}$), and by at most $(\gamma-1)$ times the
stand-alone value of~$j$.\footnote{Each additional constraint on
allowable perturbations $\tilde{f}$ weakens the stability assumption, resulting in a harder
  problem.  For example, if one only assumes~(i) and~(ii) and
  not~(iii), then the problem becomes ``too easy'', and every
  $\alpha$-approximation algorithm automatically recovers the optimal
  solution in $\alpha$-stable instances.  If~(iii) is replaced by the
  stronger condition that all marginal values change by a factor in
  $[1,\gamma]$, the problem becomes ``too hard'', with no positive
  recovery results possible (essentially because zero marginal values
  in $f$ must stay zero in $\tilde{f}$).}  
This definition specializes to the usual one in
the special case of additive functions.  
Towards the end of this section, we also present an application for the welfare maximization problem, for which a weaker stability assumption is sufficient to guarantee recovery of the optimal allocation.

\hyperref[sec:local search]{Section~\ref{sec:local search}} identifies the smallest $\gamma$ such
that all local optima of $\gamma$-stable instances are also global
optima, with both additive and submodular functions.  A byproduct of
our results here is new tight worst-case approximation guarantees for
local search in $p$-extendible systems, which surprisingly were not
known previously.  The tight approximation guarantees are $p^2$ for
additive functions and $p^2+1$ for monotone submodular functions.

\subsection{Further Related Work}

Perturbation-stability was defined by
Bilu and Linial~\cite{bilu2012stable} 
in the context of the \MaxCut\
problem.  Subsequent work on perturbation-stability
includes~\cite{bilu2012practically,makarychevSODA14,angelidakisSTOC17,balcan2012clustering,balcan2015k,awasthi2012center,reyzin2012data,mihalak2011complexity,balcan2010nash}.
Independently of
Bilu and Linial \cite{bilu2012stable}, Balcan, Blum and
Gupta~\cite{balcan2009approximate} introduced the related notion of
\textit{approximation stability} in the context of clustering problems
like $k$-means and $k$-median.  
More technically distant analogs of these stability conditions (but
with similar motivation) were proposed by~\cite{ackerman2009clusterability,daniely2012clustering,ostrovsky2006effectiveness}; see
Ben-David~\cite{ben2015computational} for further discussion.

\section{Preliminaries}\label{sec:preliminaries}
In this section we describe the notation and definitions which we use through the rest of the paper. We start by defining the family of $p$-systems and the problem of submodular maximization; then we present our two protagonist algorithms and the standard (additive) stability definition.

\begin{itemize}
\item
$p$-Systems~\cite{korte1978analysis,lee2009submodular}: Suppose we are given a (finite) ground set $X$ of $m$ elements (this could be the set of edges in a graph) and we are also given an \textit{independence family} $\I\sse 2^X$, a family of subsets that is downward closed; that is, $A \in \I$ and $B\sse A$ imply that $B\in \I$. A set $A$ is independent iff $A\in \I$. For a set $Y\sse X$, a set $J$ is called a \textit{base} of $Y$, if $J$  is a maximal independent subset of $Y$; in other words $J\in \I$ and for each $e\in Y\sm J$, $J+e\not\in \I$. Note that $Y$ may have multiple bases and that a base of $Y$ may not be a base of a superset of $Y$. $(X,\I)$ is said to be a $p$-system if for each $Y\sse X$ the following holds:
\[ \dfrac{\max_{J: J\ \text{is a base of}\ Y} |J|}{\min_{J: J\ \text{is a base of}\ Y}|J|} \le p\]
 All set systems are assumed to be down-closed. There are some interesting special cases of $p$-systems~\cite{mestre2006greedy, calinescu2011maximizing}: 
\[
\text{intersection of $p$ matroids $\sse$ $p$-circuit-bounded systems $\sse$ $p$-extendible systems $\sse$ $p$-systems} 
\]
\item $p$-extendible: An independence system $(X,\I)$ is $p$-extendible if the following holds: suppose we have $A\sse B, A,B \in \I$ and $A+e\in \I$; then there should exist a set $Z\sse B\sm A$ such that $|Z|\le p$ and $B\sm Z +e \in \I$. We note here that $p$-extendible systems make sense only for integer values of $p$, whereas $p$-systems can have $p$ being fractional and that 1-systems as well as 1-extendible systems are exactly matroids. It is a family of independence systems containing many important and seemingly unrelated problems like welfare maximization, $k$-dimensional Matching, Asymmetric Travelling Salesman Problem, weighted $\Delta$-Independent Set ($\Delta$: maximum degree) and others~\cite{mestre2006greedy}.

\item
Submodular Maximization: A set function $f : 2^X\rightarrow \RR^+\cup \{0\}$ is submodular if for every $A, B \subseteq X$, we have $f(A) + f(B) \ge f(A \cup B) +
f(A\cap B)$. Given a $p$-system $(X,\I)$ and a monotone submodular function $f$, we are interested in the problem of maximizing $f(S)$ over the
independent sets $S \in \I$; in other words we wish to find $\max_{S\in \I} f(S)$. If $f$ is additive, we can associate a weight $w_e$ with each element $e \in X$ and we want to find $\max_{S\in \I}w(S)$, where $w(S)=\sum_{e\in S}w_e$.

\item
Greedy algorithm: It starts with $S=\emptyset$ and greedily picks elements of $X$ that will increase its objective value by the most, while remaining feasible i.e. picks $e^*=\arg\max_{e\in X, S+e\in \I}(f(S+e)-f(S))$. It is a well-known fact~\cite{korte1978analysis,nemhauser1978analysis}, that for any $p$-system, the standard greedy algorithm is a ($p+1$)-approximation (if $f$ is additive, Greedy is a $p$-approximation).

\item
$(p,q)$-Local Search: It starts from a feasible solution and at each iteration seeks for an improving move. In particular, starting from any $S\in \I$, it tries to find a \textit{better} $S'\in\I$ with: $|S\sm S'|\le p$, $|S'\sm S|\le q$ and $f(S')>f(S)$. If it finds such a feasible solution $S'$, it switches to $S'$ and repeats. It stops when no improving move can be made. Note that the stopping condition and its performance depend on the size of the $(p,q)$-\textit{neighbourhood} used. We note that $(p,1)$-local search is necessary for $p$-extendible systems. For recent improvements on Local Search performance in the case of matroids, we refer the reader to \cite{lee2009submodular}.

\item
Stable instances: Stability can be defined in general for instances of weighted optimization problems~\cite{bilu2012stable}, where the objective function $w$ is additive. In our case, given a $p$-system and an \textit{additive} function $w$ we wish to maximize over the $p$-system, we call the instance $\gamma$-stable, if the optimal solution $S^* \in \mathcal{I}$ remains the unique optimum, even after assigning a new weight $\tilde{w}_e$ to an element $e$ such that $w_e\le \tilde{w}_e\le \gamma \cdot w_e$. In an extreme case, we can keep the weights of the elements in optimum the same and increase all others by a factor of $\gamma$; the optimum should remain the same. Sometimes, we say that we $\gamma$-perturb the input when we multiply some weights by at most $\gamma$. We will see in \hyperref[sec:submodular] {Section \ref{sec:submodular}} how to extend this additive stability definition to stability for submodular functions.

\end{itemize}
\section{Warm-up: Additive Case and Greedy Recovery}\label{sec:additive}
In this section, as a warm-up, we deal with additive functions, proving the first positive recovery result for the greedy algorithm and showing that it is tight.
\subsection{Exact Recovery for $p$-extendible, $p$-stable systems}

We are given an independence set system $(X,\mathcal{I},w)$ and we want to find an independent solution $S^* \in \mathcal{I}$ with maximum weight, where for $I\in\I: w(I)=\sum_{e \in I}w(e)$. We are interested in the performance of the standard greedy algorithm and we can prove the following:

\begin{theorem}\label{greedy-additive}
Given an instance of a $p$-extendible independence system $(X,\mathcal{I},w)$, that has a $p$-stable optimal solution $S^* = \arg \max_{I\in \I}w(I)$, the Greedy algorithm exactly recovers $S^*$.
\end{theorem}

\begin{proof}
From the definition of $p$-extendibility we know that for the system $\mathcal{I}$, the following holds: suppose $A\subseteq B, A,B \in \mathcal{I}$ and $A+e \in \mathcal{I}$, then there is a set $Z\subseteq B\setminus A$ such that $|Z|\le p$ and $B\setminus Z+e \in \mathcal{I}$. The Greedy starts from the empty set and greedily picks elements with maximum weight subject to being feasible; it finally outputs $S$ which is a maximal solution, i.e. $S\cup \{e\}\notin \I, \forall e\in X\sm S$. In order to get exact recovery, we want to show that $S\equiv S^*$. 

Let's suppose $S\sm S^*\neq \emptyset$. Then, out of all the elements of $S\sm S^*$ that the Greedy selected, let's focus on the first element $e_1\in S\sm S^*$. Let $S_{\{e_1\}}$ denote the greedy solution right before it picked element $e_1$. Note that before choosing $e_1$, greedy $S_{\{e_1\}}$ was in agreement with the optimal solution, i.e. $S_{\{e_1\}}\subseteq S^*$. Since $e_1 \not \in S^*$, we can use the $p$-extendibility, where we specify $A=S_{\{e_1\}}, B=S^*, e=e_1$ ($A+e \equiv S_{\{e_1\}} + e_1 \in \mathcal{I}$, since Greedy is always feasible) and we get, following the above definition, that there exists set of elements $Z\subseteq S^*\sm A \equiv S^*\sm S_{\{e_1\}}$, with $|Z|\le p$ and $(S^*\setminus Z)\cup\{e_1\} \in \mathcal{I}$. This intuitively means that the element $e_1$ has conflicts with the elements in $Z\subseteq S^* \setminus S_{\{e_1\}}$, but if we remove at most $|Z|\le p$ elements from $S^*\sm S_{\{e_1\}}$, we get no conflicts and thus an independent (feasible) solution according to the system $\mathcal{I}$. 

We call this solution $J$, i.e. $J=(S^*\setminus Z)\cup\{e_1\} \in \mathcal{I}$ (note $J\neq S^*$) and we will show that we can perturb the instance (new weight function $\tilde{w}$) no more than a factor of $p$, so that $J$'s weight is at least that of the optimal, i.e. $\tilde{w}(J)\ge \tilde{w}(S^*)$, which would be a contradiction to the $p$-stability of the given instance. All we have to do is perturb the instance by multiplying the weight of the element $e_1$ by $p$. By the greedy criterion for picking elements (note that all elements of $Z$ were available to Greedy at the point it chose $e_1$) and the fact that $|Z|\le p$ we get:
\begin{align} \label{eqeq1}
\forall e \in Z\subseteq \left(S^*\sm S_{\{e_1\}}\right): w(e_1) \ge w(e) \implies p\cdot w(e_1)\ge \sum_{e\in Z}w(e) = w(Z)
\end{align} which implies that the weight of the set $J$ is actually no less than the weight of $S^*$ in the aforementioned perturbed instance (weight function $\tilde{w}$). Indeed:
\[
\tilde{w}(J)=\tilde{w}(\left(S^*\setminus Z\right) \cup e_1) = \tilde{w}(S^*\setminus Z) + \tilde{w}(e_1) = w(S^*)-w(Z)+p\cdot w(e_1)\ge w(S^*)=\tilde{w}(S^*)
\]
where for the last inequality we used (\ref{eqeq1}). This is a contradiction because it violates the $p$-stability property (the optimal solution should stand out as the unique optimum for any $p$-perturbation) and thus we conclude that $S\sm S^*=\emptyset$. Since Greedy outputs a maximal solution, we conclude that $S$ coincides with $S^*$ and so Greedy exactly recovers the optimal solution.
\end{proof}

We next show that our result is tight both in terms of the stability factor and the generality of $p$-extendible systems. 

\begin{proposition}
There exist $p$-extendible systems with a $(p-\epsilon)$-stable optimal solution $S^*$, for which the Greedy fails to recover it.
\end{proposition}

\begin{proof}
Take a Maximum Weight Matching instance (here $p=2$): a path of length
3 with weights (1,$1+\epsilon'$,1). The Greedy fails to recover the
optimal solution $S^*$, since it picks the $(1+\epsilon') $ edge
whereas it should have picked both the other edges. For the right
choice of $\epsilon'$ ($\epsilon'<\tfrac{\epsilon}{2-\epsilon}$), we
can make the instance arbitrarily close to $(p-\epsilon)=(2-\epsilon)$
stable.  Observe that we can give such examples for any value of $p$
(consider the $p$-dimensional Matching problem) and that the example
can be made arbitrarily large just by repeating it.  
\end{proof}

\begin{proposition} \label{knapsack}
There are $p$-systems whose optimal solution $S^*$ is $M$-stable (for arbitrary $M>1$) and for which the greedy algorithm fails to recover it.
\end{proposition}

\begin{proof}
The example is based on a knapsack constraint. Fix $M'>1$ and let the size of the knapsack $B=M'+1$. We will have elements of type $A$ ($|A|=M'$), a special element $e^*$ and elements of type $C$ ($|C|=M'$). The pair (value, size) for elements in $A,C$ is respectively: $(2,1),(1,\tfrac{1}{M'})$ and for $e^*: (1+\epsilon,1), \epsilon>0$. Note that the optimal solution $S^*$ is $A\cup C$ with total value $2M'+M'=3M'$ and size $M'+M'\cdot\tfrac{1}{M'}=M'+1$ (fits in the knapsack). However, Greedy will pick $A\cup\{e^*\}$ for a total value of $2M'+1+\epsilon$ and size $M'+1$. Note that this is a $p$-system for a value of $p<2$ since any feasible solution $S$ can be extended to a solution $S'$ with $|S'|\ge M'+1$ and the largest feasible solution has $2M'$ elements (there are only $2M'+1$ elements in total). However, this is not a $2$-extendible system (it is actually an $M'$-extendible system) and we see that even if it is ($M'-1$)-stable, Greedy still fails to recover the optimal solution $S^*$. To see why it is ($M'-1$)-stable, note that the only $\gamma$-perturbation (perturbations are allowed only on the values, not the sizes) we can make to favour the greedy solution is to the element $e^*$, thus we would need $\gamma(1+\epsilon)\ge M' \implies \gamma>M'-1$ ($\epsilon$ is small). Choose $M'=M+1$ and this concludes the proof. We also note that a variation of this counterexample would trick as well the (more natural) Greedy that sorts the elements according to value density ($\tfrac{v_i}{s_i}$) instead of just their value.
\end{proof}

We find \hyperref[knapsack]{Proposition~\ref{knapsack}} surprising, given that the greedy
algorithm is a good worst-case approximation algorithm for such
problems. The above ``bad'' example leads us to the definition of \textit{hereditary} systems; it turns out that this is another characterization of the $p$-extendible systems. Due to space constraints, we defer the definition until \hyperref[sec:hereditary]{Appendix~\ref{sec:hereditary}}.





\section{The Case of Submodular Functions}\label{sec:submodular}

This section considers recovery results for stable instances where the
objective function is monotone and submodular. Submodular functions
are widely used in many areas ranging from mathematics to economics,
and they model situations with \textit{diminishing
  returns}. Famous examples include influence
maximization~\cite{kempe2003maximizing,mossel2007submodularity,
  easley2010networks, chen2009efficient} and welfare maximization in
auctions and game theory~\cite{lehmann2001combinatorial,
  nisan2007computationally}. For example, in influence maximization,
the goal is to ``activate'' a subset of the participants in a social 
network (e.g., provide with information, or a promotional product) so as to
maximize the expected spread of the idea or product.
The diffusion of information is usually modeled with submodular
functions (indicating the probability that a node adopts a new idea or
product as a function of how many of her neighbors in the social
network have already done so).
In practice, the submodular functions in the input are estimated from
data and hence are noisy (e.g.~\cite{backstrom2006group}).  One hopes that the output of an influence
maximization algorithm (which is typically a greedy
algorithm~\cite{kempe2003maximizing}) is robust to modest errors in
the specification of the submodular function.  This section proposes a
definition to make this idea precise, and proves tight results for greedy
and local search algorithms under this stability notion.


\subsection{Stability for submodular functions} \label{submod-def}

All previous work on perturbation-stability considered only additive
objective functions.  We next state our extension to submodular functions.

\begin{definition} [$\gamma$-perturbation, $\gamma \ge 1$]\label{d:sm} \label{def:stability}

Given a monotone submodular function $f: 2^X\rightarrow \RR^+\cup \{0\}$, we define $f_S(j) = f(S+j) - f(S)$.
A $\gamma$-perturbation of $f$ is any function $\tilde{f}$ such that the following three properties hold:
\begin{enumerate}
\item $\tilde{f}$ is monotone and submodular.
\item $f \le \tilde{f} \le \gamma f$, or in other words $f(S) \le \tilde{f}(S)\le \gamma f(S)$ for all $S \subseteq X$.
\item For all $S \subseteq X$ and $j \in X \setminus S$, $0 \le \tilde{f}_S(j) - f_S(j) \le (\gamma -1)\cdot f(\{j\})$. \label{property}
\end{enumerate}
\end{definition}
The definition of a $\gamma$-stable instance is then defined as usual.

\begin{definition} [$\gamma$-stability]\label{d:sm2}
  Given an independence system $(X,\I)$ and a monotone submodular
  function $f: 2^X\rightarrow \RR^+\cup \{0\}$, let $S^*:=\arg\max_{S\in \I}f(S)$. The instance is
  {\em $\gamma$-stable} if
  for every $\gamma$-perturbation of the initial function $f$, $S^*$
  remains the unique optimal solution.
\end{definition}
As discussed in the \hyperref[sec:introduction]{Introduction}, 
while \hyperref[d:sm]{Definition~\ref{d:sm}} is perhaps not the first one that comes to
mind, it appears to be the ``sweet spot''.  Natural modifications\footnote{If we dropped Property 3, then \textit{any} $c$-approximation algorithm ($c\ge1$) returning a solution $S$ with $f(S)\ge\tfrac{1}{c}\cdot f(S^*)$, could be made to have value equal with $S^*$ in the $c$-perturbed version $\tilde{f}(S)=cf(S)\ge f(S^*)=\tilde{f}(S^*)$. If we dropped Property 2, then the definition would not be a generalization for the case of additive perturbations as we could have $\tilde{f}(S)> \gamma f(S)$ for some set $S$, because of the quantity $f(\{j\})$ in Property 3, which is relative to the \textit{empty set} and may be large compared to $f_S(j)$.} of
the definition are generally either too restrictive (rendering the
problem impossible, e.g.~if property~3 is replaced with relative
perturbations since then the zero marginal values in $f$ must stay zero in $\tilde{f}$) or too permissive (rendering the problem uninteresting,
with all $\alpha$-approximation algorithms equally good).

\begin{proposition}
  \hyperref[d:sm]{Definition~\ref{d:sm}} specializes to perturbation-stability in the
  special case of an additive objective function.
\end{proposition}
\begin{proof}
This follows easily since if the function $f$ is additive, then there will be no dependence of the element's $j$ marginal value on the current set $S$ and thus property 3 from the above $\gamma$\textbf{-perturbation} definition just becomes:
$$
0 \le \tilde{f}_S(j) - f_S(j) \le (\gamma -1)\cdot f(j) \iff 0\le \tilde{f}(j)-f(j) \le  (\gamma -1)\cdot f(j) \iff f(j)\le  \tilde{f}(j) \le \gamma \cdot f(j)
$$
which is exactly the standard notion of stability introduced by \cite{bilu2012stable}. Note that this also implies the first condition for all sets $S$: $f(S) \le \tilde{f}(S)\le \gamma f(S)$, by the additivity of $f$.
\end{proof}

We now prove a useful proposition that we will often use when proving recovery results for submodular maximization. Informally, we show that multiplying the marginal improvements of the choices made by an algorithm by $\gamma$ is a valid $\gamma$-perturbation. 

\begin{proposition}
\label{prop:perturb}
Let $f$ be a monotone submodular function.
Fix an ordered sequence of elements $e_1,e_2,\dots,e_k$, and let $\delta_i = f(\{e_1,\ldots,e_i\}) - f(\{e_1,\ldots,e_{i-1}\})$.
Then $\tilde{f}$ defined by $\tilde{f}(S) := f(S) + (\gamma-1) \sum_{i: e_i \in S} \delta_i $ is a valid $\gamma$-perturbation of $f$.
\end{proposition}

\begin{proof}
Let us verify the conditions of a $\gamma$-perturbation.

First, $\tilde{f}$ is monotone submodular, since it is a sum of a monotone submodular and a monotone additive function ($\delta_i \geq 0$ by monotonicity).

Second, we have $f(S) \leq \tilde{f}(S) = f(S) + (\gamma-1) \sum_{i: e_i \in S} \delta_i \leq f(S) + (\gamma-1) f(S \cap \{e_1,\ldots,e_k\})$ by submodularity, and by monotonicity this is at most $\gamma f(S)$.

Third, the marginal values of $\tilde{f}$ are $\tilde{f}_S(e_i) = f_S(e_i) + (\gamma-1) \delta_i \leq f_S(e_i) + (\gamma-1) f(\{e_i\})$ (and unchanged for elements other than the $e_i$). 
\end{proof}

\subsection{Greedy recovery and submodularity}

The main result here is that the standard greedy algorithm can recover the optimal solution of a $p$-extendible system, if the optimal solution is $(p+1)$-stable (as it was defined in \hyperref[submod-def]{Section~\ref{submod-def}}). 

\begin{theorem}[Greedy Recovery] \label{th:submod}
Given a monotone submodular function $f$ to maximize over a $p$-extendible system $(X,\I)$, if the optimal solution $S^*=\arg\max_{S\in \I}f(S)$ is $(p+1)$-stable, then the greedy algorithm recovers $S^*$ exactly.
\end{theorem}

\begin{proof}
The proof generalizes the argument we used in the additive case so that we handle submodularity and the proving strategy resembles the proof of the approximation guarantee for the greedy algorithm for submodular maximization on $p$-extendible systems~\cite{calinescu2011maximizing}.  Let's denote by $S = \{e_1,\ldots,e_k\}$ the solution produced by Greedy (in the order that Greedy picked them) and $S^*$ the optimal solution. To give some intuition, in the additive case before, we used the property of $p$-extendibility in order to say that every element that appears in $S$ but not in $S^*$ could be ``boosted'' by a factor of $p$ to obtain an even better optimal solution, which would be a contradiction, because of the $p$-stability. Now, due to submodularity, we need to be careful that we make this exchange argument in a cautious manner.

For $0 \leq i \leq k$, let $S_i = \{e_1,e_2,\dots, e_i\}$ denote the first $i$ elements picked by Greedy (with $S_0=\emptyset$). Let $\delta_i = f_{S_{i-1}}(e_i)= f(S_i)-f(S_{i-1})$. 
Using the $p$-extendibility property, we can find a chain of sets $S^* = T_0 \supseteq T_1\supseteq \dots \supseteq T_k = \emptyset$ such that for $1\le i\le k$:
$$S_i\cup T_i \in \I,  \ S_i\cap T_i = \emptyset \ \text{and}\ |T_{i-1} \setminus T_{i}| \leq p.$$
The above means that every element in $T_i$ is a candidate for Greedy in step $i+1$. We construct the chain as follows: Let $T_0=S^*$; we show how to construct $T_i$ from $T_{i-1}$:

\begin{enumerate}
\item If $e_i \in T_{i-1},$ we define $S^*_i=\{e_i\}$ and $T_i=T_{i-1}-e_i$. This corresponds to the trivial case when Greedy, at stage $i$, happens to choose an element $e_i$ that also belongs to the optimal solution $S^*$.
\item Otherwise ($e_i \notin T_{i-1}$), we let $S^*_i$ be a smallest subset of $T_{i-1}$ such that $(S_{i-1}\cup T_{i-1})\sm S^*_i + e_i$ is independent and since $\I$ is $p$-extendible, we have $|S^*_i|\le p$. We let $T_i=T_{i-1}\sm S^*_i$. 
\end{enumerate}

By the above definitions for $S_i,T_i,S^*_i$ it follows that $S_i\cup T_i \in \I$ and $S_i \cap T_i=\emptyset$. By the maximality of Greedy (stopping condition: $\{e| S_k+e \in \I\}=\emptyset$) and the fact that $S_k\cup T_k \in \I$, it also follows that $T_k=\emptyset$. Since Greedy could have picked, instead of $e_i$, any of the elements in $S^*_i$ (in fact $T_{i-1}$) we get: $\delta_i \ge \tfrac{1}{p}f_{S_{i-1}}(S^*_i)$ (recall that $|S^*_i|\le p$).

Let us assume now that the Greedy solution $S$ is not optimal.
We use \hyperref[prop:perturb]{Proposition~\ref{prop:perturb}} to define a $(p+1)$-perturbation that produces a new optimal solution.
Let's suppose $|S\sm S^*|=l$ and let's rename the elements $e_i$ such that $|S \sm S^*|=\{e_1,e_2,\dots, e_l\}$ in the order that the Greedy picked the elements.
Then we define $\tilde{f}(T)$ for every $T$ by $f(T) = f(T) + p \sum_{1 \leq i \leq l: e_i \in T} \delta_i$, where $\delta_i = f_{S_i}(e_i) = f(\{e_1,\ldots,e_i\}) - f(\{e_1,\ldots,e_{i-1})$. Using \hyperref[prop:perturb]{Proposition~\ref{prop:perturb}}, this is a valid $(p+1)$-perturbation.
For the greedy solution $S$, we obtain:
\[ \tilde{f}(S) =  f(S) + p\sum_{i=0}^{l-1} f_{S_i}(e_{i+1}) \ge\]
\[ \ge f(S) + \sum_{i=0}^{l-1}f_{S_i}(S^*_{i+1}) \ge 
f(S) + \sum_{i=0}^{l-1}f_S(S^*_{i+1})\ge f(S) + f_S(S^*\sm S)= \]
\[=f(S) + \left( f((S^*\sm S)\cup S)- f(S)\right)=f(S^*\cup S)\ge f(S^*) = \tilde{f}(S^*).\]
We ended up with $\tilde{f}(S) \ge \tilde{f}(S^*)$ which means that $S^*$ is no longer the unique optimum and hence we get a contradiction to the $(p+1)$-stability of $S^*$.
\end{proof}

\begin{remark}
If instead of exact access to the values of the function $f$, we had an $\alpha$-approximate oracle, then the proof easily extends to handle this case as well. In particular, suppose each element $e_i$ picked by Greedy at stage $i$ satisfies $f_{S_{i-1}}(e_i)\ge \alpha \max_{e\in A_i} f_{S_{i-1}}(e)$, where $A_i$ is the set of all candidate augmentations of $S_{i-1}$. Here $\alpha\le 1$. We would then have that the greedy marginal improvement $\delta_i\ge \tfrac{\alpha}{p}f_{S_{i-1}}(S^*_i)$ and thus we would need $\gamma-1=\tfrac{p}{\alpha}$ leading to exact recovery of $\left(\tfrac{p+\alpha}{\alpha}\right)$-stable instances ($\alpha\le 1$).
\end{remark}

\subsection{Welfare Maximization}
In many situations, like the welfare maximization problem~\cite{lehmann2001combinatorial,nisan2007computationally,vondrak2008optimal}, the submodular function $f$ we wish to maximize has a special form, e.g. it may be written as a sum of other submodular functions $f_i$ (each of which may correspond to the player's $i$ valuation on different allocations of the items). In this special case, we have $f(S)=\sum_{i=1}^nf_i(S)$ and from \hyperref[th:submod]{Theorem~\ref{th:submod}} Greedy recovers the optimal solution $S^*$ for the case of matroids, which are 1-extendible, if $S^*$ is 2-stable. 

However, for \textit{sum} functions $f=\sum_if_i$, we may as well hope that a stronger recovery result is true, i.e. that Greedy recovers the optimal solution of $\max\{f(S)=\sum_if_i: S\in \I\}$, where the optimum is 2-stable only with respect to $2$-perturbations of the individual functions $f_i$. This is indeed true (for the proof, we refer the reader to \hyperref[app:submod]{Appendix~\ref{app:submod}}).

\begin{theorem} \label{th:submod}
Let $(X,\I)$ be a matroid on the elements of $X$, let $B_1,B_2,\dots,B_k$ be a partition of $X$, $f_i: 2^{B_i}\to \RR^+\cup\{0\}$, for $i\in \{1,2,\dots,k\}$ be monotone submodular and let $f=\sum_{i=1}^kf_i$. Suppose the optimal solution $S^*$ of $\max\{f(S): S\in \I\}$ is $2$-stable only with respect to individual perturbations of the functions $f_i$. Then, Greedy recovers $S^*$.
\end{theorem}




\section{Local Search Performance}\label{sec:local search}

In this section we discuss {\em local search}~\cite{lenstra2003local} (described in \hyperref[sec:preliminaries]{Section~\ref{sec:preliminaries}}).
Local search often gives better results than Greedy, at the cost of a slower running time --- for example for submodular maximization subject to the intersection of $k$ matroids \cite{lee2009submodular,filmus2012power}, and for $k$-set packing \cite{SviridenkoW13,Cygan13,FurerY14}. For some interesting recent results about local search in \textit{beyond-worst-case} settings and on geometric optimization we refer the reader to~\cite{cohen2016local,cohen2014unreasonable,cohen17}.

Somewhat surprisingly, it was not known (to our knowledge) how local search performs for $p$-systems and $p$-extendible systems. (We recall that the greedy algorithm gives a factor of $1/p$ for maximization of an additive function and $1/(p+1)$ for maximization of a monotone submodular function under these constraints.)
Here, we prove that local search in fact performs worse than Greedy for these constraints. Although it gives a $1/p$-approximation for cardinality maximization under a $p$-system constraint (essentially by definition), it does not give any bounded approximation factor for additive function maximization under a $p$-system, and only a $1/p^2$-approximation under a $p$-extendible system.

\subsection{Local search fails for $p$-systems}
We construct simple examples where local search will not recover any fraction of the maximum-weight solution for $p$-systems (even if it is arbitrarily stable, $p=2$, and even if we allow large exchange neighborhoods). In particular, consider a ground set $X = A \cup \{e^*\}$ where $|A|=n$. The independent sets of $\I$ are:
\begin{itemize}
\item any subset of $A$, or
\item $e^*$ plus any subset of at most $n/2$ elements of A. 
\end{itemize}

Note that this is a 2-system, because for $S \subseteq X$, any independent subset of $S$ can be extended to an independent set of size at least $\min \{|S|, n/2\}$, and the maximum independent subset of $S$ has size at most $\min \{|S|,n\}$. The weights could be 0 on $A$, and 1 on the special element $e^*$. So the optimum is $w(e^*)$ = 1 (observe that the optimal solution is $c$-stable for arbitrarily large $c$). However, $A$ is a local optimum, unless we are willing to swap out $n/2$ elements, which is not possible for efficient local search.

\subsection{Lower bound for $p$-extendible systems}

Let us consider the following instance. Let $X = A \cup B$ where $A, B$ are disjoint sets. We define $\I \subseteq 2^X$ as follows: $S \in \I$ iff
\begin{itemize}
\item $|S \cap A| + p |S \cap B| \leq |A|$, or
\item $p|S \cap A| + |S \cap B| \leq |B|$.
\end{itemize}

\begin{lemma}
For any $A,B$ disjoint, the above is a $p$-extendible system.
\end{lemma}

\begin{proof}
Let $S \subseteq T$ and $i \in X \setminus T$ be such that $S+i \in \I$ and $T \in \I$. We need to prove that there is $Z \subseteq T \setminus S, |Z| \leq p$ such that $(T \setminus Z) + i \in \I$. We can assume that $|T \setminus S| > p$, because otherwise we can set $Z = T \setminus S$ and obviously $(T \setminus Z) + i = S + i \in \I$.
Assuming $|T \setminus S| > p$, let $Z$ be an arbitrary set of $p$ elements from $T \setminus S$. We consider 2 cases: If $|T \cap A| + p |T \cap B| \leq |A|$, then $|(T \setminus Z) \cap A| + p |(T \setminus Z) \cap B| \leq |A| - p$. Adding the element $i$ can increase the left-hand side by at most $p$, and so $|(T \setminus Z + i) \cap A| + p |(T \setminus Z + i) \cap B| \leq |A|$. Similarly, in the second case, if $p |T \cap A| + |T \cap B| \leq |B|$, then $p |(T \setminus Z) \cap A| + |(T \setminus Z) \cap B| \leq |B| - p$. Adding the element $i$ can increase the left-hand side by at most $p$, and so $p |(T \setminus Z + i) \cap A| + |(T \setminus Z + i) \cap B| \leq |B|$. 
\end{proof}

Now we choose the cardinalities of $A$ and $B$ and the weights of their elements appropriately to get a negative result.

\begin{lemma}
For $\epsilon>0$, let $|A| = n$ and $|B| = (p - \epsilon) n$, and set the weights as $w_a = 1$ for $a \in A$ and $w_b = p - \epsilon$ for $b \in B$.
Then $A$ is a local optimum of value $w(A) = w(B) / (p - \epsilon)^2$, unless the local search explores exchanges of size at least $\frac{\epsilon}{p} n$.
\end{lemma}

\begin{proof}
Both $A$ and $B$ are independent sets. 
Note that for any $i \in B$, we need to remove $Z \subseteq A$ of cardinality at least $|Z| = p$ to obtain $S = (A \setminus Z) + i$ satisfying $|S \cap A| + p|S \cap B| \leq |A|$. More generally, for $Y \subseteq B$, we need to remove $Z \subseteq A, |Z| = p|Y|$ to obtain $S = (A \setminus Z) \cup Y$ that satisfies $|S \cap A| + p|S \cap B| \leq |A|$. Possibly, we could satisfy the second condition, $p|S \cap A| + |S \cap B| \leq |B|$, but this will not happen unless $|A \setminus Z| = |S \cap A| \leq |B| / p = (1 - \frac{\epsilon}{p}) n$. Therefore, we would need to remove $Z$ of cardinality at least $\frac{\epsilon}{p} n$. If the swaps considered are smaller than $\frac{\epsilon}{p} n$ then $A$ is a local optimum because adding $Y \subseteq B$ and removing $Z \subseteq A, |Z| = p |Y|$ results in a solution of lower weight. In conclusion, $A$ is a local optimum of value $w(A) = n$, while the optimum is $OPT = w(B) = (p-\epsilon)^2 n$.
\end{proof}

\subsection{Upper bound for $p$-extendible systems}

Here we prove that local search does in fact provide a $1/p^2$-approximation for weighted maximization under a $p$-extendible system. More generally, we prove (here, we will ignore the technicalities of stopping the local search in polynomial time as this can be handled using standard techniques, while losing $1/poly(n)$ in the approximation factor) the following:

\begin{theorem} \label{th:LS-approx}
For any $p$-extendible system $\I \subseteq 2^X$ and a monotone submodular function $f:2^X \rightarrow \RR_+$,
local search with $(p,1)$-swaps (including at most $1$ element and removing at most $p$ elements) provides a $1/(p^2+1)$-approximation. For additive $f$, the factor is $1/p^2$.
\end{theorem}

\begin{proof}
Let $A$ be a local optimum under $(p,1)$-swaps, and let $B$ be an optimal solution. (For convenience, let us also assume that we always try to add elements to $A$ if possible, even if they bring zero marginal value.) We proceed in two steps, the first one inspired by the analysis of the greedy algorithm for $p$-extendible systems \cite{calinescu2011maximizing} and the second one similar to other analyses of local search.

Let $A = \{a_1,\ldots,a_k\}$ be a greedy ordering of $A$ in the sense that $a_1$ is the element of $A$ maximizing $f_\emptyset(a_1)$; given $a_1$, $a_2$ is the element of $A-a_1$ maximizing $f_{\{a_1\}}(a_2)$, $a_3$ is the element of $A-a_1-a_2$ maximizing $f_{\{a_1,a_2\}}(a_3)$, etc. Using the $p$-extendible property, there is a subset $B_1 \subseteq B, |B_1| \leq p$ such that $(B \setminus B_1) + a_1 \in \I$. Further, since $\{a_1,a_2\} \in \I$, there is a subset $B_2 \subseteq B \setminus B_1, |B_2|\leq p$ such that $(B \setminus (B_1 \cup B_2)) \cup \{a_1,a_2\} \in \I$, etc. Generally, there are disjoint subsets $B_1,\ldots,B_k \subseteq B, |B_i| \leq p$ such that $(B \setminus (B_1 \cup \ldots B_i)) \cup \{a_1,\ldots,a_i\} \in \I$. In fact, if $|A| = k$, the sets $B_1,\ldots,B_k$ form a partition of $B$. Otherwise there would be additional elements in $B \setminus (B_1 \cup \ldots \cup B_k)$ which can be added to $A$, which would contradict the local optimality of $A$.

Now, we claim that for each $b \in B_i$, we have $f_A(b) \leq p f_{\{a_1,\ldots,a_{i-1}\}}(a_i)$. If not, we would be able to add $b$ and, since $\{a_1,\ldots,a_{i-1},b\} \in \I$, we could remove at most $p$ elements $Z \subseteq A \setminus \{a_1,\ldots,a_{i-1}\}$ so that $(A \setminus Z) + b \in \I$. By submodularity and the greedy ordering, we would have $f(A \setminus Z) \geq f(A) - p f_{\{a_1,\ldots,a_{i-1}\}}(a_i)$ and again by submodularity, we would have $f((A \setminus Z) + b) \geq f(A \setminus Z) + f_A(b) > f(A \setminus Z) + p f_{\{a_1,\ldots,a_{i-1}\}}(a_i) \geq f(A)$. Therefore, this would be an improving local swap.

Since $A$ is a local optimum, we conclude that $f_A(b) \leq p f_{\{a_1,\ldots,a_{i-1}\}}(a_i)$ for each $b \in B_i$. Since $B = B_1 \cup \ldots \cup B_k$ and $|B_i| \leq p$, we have by submodularity
$$ f_A(B) \leq \sum_{i=1}^{k} \sum_{b \in B_i} f_A(b) \leq \sum_{i=1}^{k} |B_i| p f_{\{a_1,\ldots,a_{i-1}\}}(a_i)
\leq p^2 \sum_{i=1}^{k} f_{\{a_1,\ldots,a_{i-1}\}}(a_i) \leq p^2 f(A) $$
For $f$ monotone submodular, we have $f(B) \leq f(A) + f_A(B) \leq (p^2+1) f(A)$.
For $f$ additive, we have $f(B) = f_A(B) \leq p^2 f(A)$. This completes the proof.
\end{proof}

\subsection{Recovery for $p$-extendible systems}
Note that in the proof of \hyperref[th:LS-approx]{Theorem \ref{th:LS-approx}}, if we focus on comparing the values of $A\sm B$ and $B\sm A$, we will be able to get exact recovery as we can perturb only $A\sm B$. The following theorem intuitively tells us that \textit{local optima of stable instances are global optima}.
\begin{theorem} \label{th:LS-recovery}
Given a $p$-extendible system $\I \subseteq 2^X$ and a monotone submodular function $f:2^X \rightarrow \RR_+\cup\{0\}$ we wish to maximize, if the optimal solution $B$ is $(p^2+1)$-stable, then local search with $(p,1)$-swaps exactly recovers it. If $f$ is additive, recovery holds if $B$ is $p^2$-stable.
\end{theorem}

\begin{proof}
The basic idea is that we can contract the elements that belong to $A\cap B$ and then use the same charging argument from above. Using the notation from the proof of \hyperref[th:LS-approx]{Theorem \ref{th:LS-approx}}, for elements $a_i\in A\cap B$ the corresponding $B_i$ block is just $\{a_i\}$. Now we can rename elements in $A\sm B=\{a_1,\dots,a_m\}$ with corresponding blocks $B_1,\dots,B_m$ such that $B\sm A = B_1 \cup \ldots \cup B_m$ and $|B_i| \leq p$. Rewriting the local search guarantee:
$$f_A(B\sm A) \leq \sum_{i=1}^{m} \sum_{b \in B_i} f_A(b) \leq \sum_{i=1}^{m} |B_i| p f_{\{a_1,\ldots,a_{i-1}\}}(a_i)
\leq p^2 \sum_{i=1}^{m} f_{\{a_1,\ldots,a_{i-1}\}}(a_i) \leq p^2 f(A\sm B)$$
Since $f_A(B\sm A)=f(B\cup A)-f(A)\ge f(B)-f(A)$, we can $(p^2+1)$-perturb the input (only the marginal of elements in $A\sm B$) and get: $\tilde{f}(B)=f(B) \le f(A)+p^2f(A\sm B)=\tilde{f}(A)$, hence contradicting the $(p^2+1)$-stability.
In the case of additive $f$, $f_A(B\sm A)= f(B\sm A)$ and $\tilde{f}(B)=f(B)=f(B\sm A)+f(B\cap A)\le p^2f(A\sm B) +f(B\cap A)\le f(A)+(p^2-1)f(A\sm B)=\tilde{f}(A)$, where we $p^2$-perturbed the instance, hence contradicting the $p^2$-stability of the instance.
\end{proof}


\subsection{Recovery for the intersection of Matroids}
If the independence system $\I$ is the intersection of $p$ matroids: $\I=\cap_{i=1}^p \I_i$, local search with ($p,1$)-swaps recovers $(p+1)$-stable optimal solutions (for proof, see \hyperref[app:LSrecovery]{Appendix~\ref{app:LSrecovery}}). 

\begin{theorem}\label{th:LSmatroids}
Given $(X,\I)$, with $\I=\cap_{i=1}^p \I_i$ where each $\I_i$ is a matroid and $f$ monotone submodular, such that the optimal solution is $(p+1)$-stable, Local Search exactly recovers it.
\end{theorem}

\subsection*{Acknowledgements}\label{ack}
TR was supported by NSF award CCF-1524062. We also thank the anonymous reviewers for their useful comments.

\bibliography{main}

\appendix

\section{Proof of \hyperref[th:submod]{Theorem} for welfare maximization}\label{app:submod}

\begin{theorem}
Let $(X,\I)$ be a matroid on the elements of $X$, let $B_1,B_2,\dots,B_k$ be a partition of $X$, $f_i: 2^{B_i}\to \RR^+\cup\{0\}$, for $i\in \{1,2,\dots,k\}$ be monotone submodular and let $f=\sum_{i=1}^kf_i$. If the optimal solution $S^*$ of $\max\{f(S): S\in \I\}$ is $2$-stable with respect only to individual perturbations of the functions $f_i$, greedy will recover $S^*$.
\end{theorem}

\begin{proof}
We note that the $B_i$'s form a partition of $X$ which is not tied to the matroid in any way. To avoid confusion, we should first emphasize the greedy algorithm in this case: It starts with the empty set $S_0=\emptyset$, at step $t$ it selects: $e=\arg\max_{x\in X}\{f(S_{t-1}+x)-f(S_{t-1})\}$ subject to the matroid constraint and it updates $S_t\leftarrow S_{t-1}+e$. This is a particular instantiation of the standard greedy algorithm in welfare maximization that first picks an item giving it to the player so that it yields maximum marginal improvement.

Suppose greedy outputs $S\neq S^*$ and that it chose elements $A_i\subseteq B_i$. Let $S_e$ be the greedy solution right before adding element $e$. Then a 2-perturbation of the individual functions is:
\[
\tilde{f_i}(A_i)=f_i(A_i)+\sum_{e\in (S\cap B_i)\sm S^*}f_i((B_i\cap S_e)+e)-f_i(B_i\cap S_e)
\]
Now coming back to the total welfare function $f$ we get:
\[
\tilde{f}(S)=\sum_{i=1}^k\tilde{f}_i(S\cap B_i)=f(S)+\sum_{e\in S\sm S^*}f_{S_e}(e)\ge f(S)+\sum_{e'\in S^*\sm S, e\leftrightarrow e'}f_{S_e}(e')\ge
\]
\[
\ge f(S)+f_S(S^*\sm S)\ge f(S^*\cup S)\ge f(S^*)=\tilde{f}(S^*)
\]
where we made use of the greedy criterion, submodularity and the matroid matching $e\leftrightarrow e'$ between elements $e\in S\sm S^*$ and $e' \in S^*\sm S$. We got $\tilde{f}(S)\ge \tilde{f}(S^*)$, hence a contradiction to the 2-stability of $S^*$ and hence $S\equiv S^*$ and greedy exactly recovers the optimal solution. 
\end{proof}


\section{Hereditary Systems}\label{sec:hereditary}

Motivated by the ``bad'' example (see \hyperref[knapsack]{Proposition \ref{knapsack}}) for the greedy algorithm, we define a new notion of an independence system that we call \textit{hereditary} $p$-system or $p$-\textit{hereditary} that as we see later is a different characterization of $p$-extendible systems. In the aforementioned example, even though we started with a $p$-system, as we progressed picking elements with the greedy algorithm, the system became a $p'$-system with $p' \gg p$, thus leading to bad performance for the greedy, even though we had the optimal solution being stable by a large amount.

The intuition behind the following definition is that we want our system to remain a $p$-system under deletions and contractions of elements.

\begin{definition}[Hereditary $p$-system]
A $p$-system $(X,\I)$ is said to be \textit{hereditary} if:
\begin{enumerate}
\item For each set $Y\sse X$, the system $(X', \I | X')$\footnote{By $(X',\I | X')$, we mean the \textit{restriction} of $\I$ to the set of elements $X'$, which is the independence system on the set $X'$, whose independent sets are the independent sets of the initial set $\I$ that are contained in $X'$.}, where $X'=X \sm Y$, is a $p$-system. This corresponds to the {\bf deletion} of the elements in $Y$ from the system.
\item For each set $Y\sse X$,  the system $(X\sm Y, \I / Y)$\footnote{By $(X\sm Y,\I / Y)$, we mean the \textit{contraction} of $\I$ by $Y$, which is the independence system on the underlying set $X\sm Y$, whose independent sets are the sets $Z \sse X\sm Y$, such that $Z\cup Y \in \I$.} is a $p$-system. This corresponds to the {\bf contraction} of the elements in $Y$.
\end{enumerate}
\end{definition}

\noindent Looking back at our ``bad'' Knapsack example we see that it is not a hereditary system since initially $p=\tfrac{2M}{M+1}\le 2$, but after we had picked all the elements in set $A$, the system on the remaining elements became an $M$-extendible system. We now prove that the family of hereditary $p$-systems coincides with the family of $p$-extendible systems.  

\begin{proposition}
A $p$-system is $p$-hereditary if and only if it is $p$-extendible.
\end{proposition}

\begin{proof}
$p$-hereditary $\implies$ $p$-extendible: Let's first think of $p$ as an integer; as we will see afterwards only this case (with integer $p$) is interesting. Suppose we had a $p$-hereditary system that was not $p$-extendible. By negating the definition of $p$-extendibility (see \hyperref[sec:preliminaries]{Preliminaries}), it follows that there exist sets $A,B \sse X$ with $A \sse B$, $A,B \in \I$ and $A\cup\{e\} \in \I$ such that for all sets $Z\sse B\sm A$ with $|Z| \le p$: $(B\sm Z)\cup \{e\} \not\in \I$. Define $Z_0 \sse B\sm A$ to be the smallest set that we need to remove from $B$ in order to have: $(B\sm Z_0)\cup \{e\} \in \I$. We know that $|Z_0|>p$ and thus, by the hereditary property, if we project the independence system on the elements $Z_0\cup \{e\}$, we get $Z_0\cup \{e\} \not\in \I$ with the ratio $\tfrac{|Z_0|}{|\{e\}|}=\tfrac{|Z_0|}{1}>p$, which contradicts the fact that we started with a $p$-hereditary system.

For $p$-extendible $\implies$ $p$-hereditary: This direction follows easily just by the definition of $p$-extendibility. To handle non-integer values of $p$, we observe that by the first argument above, a $p$-hereditary system is actually $\lfloor p\rfloor$-extendible and thus, it is $\lfloor p\rfloor$-hereditary (e.g. a 2.9-hereditary system is 2-extendible).
\end{proof}

\section{Proof of \hyperref[th:LSmatroids]{Theorem} for Intersection of Matroids and recovery}\label{app:LSrecovery}

\begin{theorem}
Given $(X,\I)$, with $\I=\cap_{i=1}^p \I_i$ where each $\I_i$ is a matroid and $f$ monotone submodular, such that the optimal solution is $(p+1)$-stable, Local Search exactly recovers it.
\end{theorem}
\begin{proof}
We denote with $A$ our local search solution (let it be maximal, even if new elements add zero value to it) and with $B$ the global optimum. Let $Y=B\sm A=\{y_1,y_2,\dots,y_k\}$ be the elements of the optimum that local search didn't choose. By the matching property~\cite{reichel2007evolutionary} of the matroids we get:
\[
\exists\ X^1,X^2,\dots,X^p \subseteq A\sm B, \text{where\ } X^j=\{x_1^j,x_2^j,\dots,x_k^j\}\  \text{such that:\ }
\]
\[
\forall j\in\{1,\dots,p\}: x_i^j \in C_j(A,y_i), \forall i\in \{1,2,\dots,k\},
\]
where $C_j(A,y_i)$ is the circuit (minimally dependent set) created in matroid $\I_j$ when adding $y_i$ in $A$.

Using the local search (with $(p,1)$ swaps) stopping condition, we have: (for ease, we use $+,-$ instead of the more accurate $\cup, \sm$)
\[
f(A+y_i-x_i^1-x_i^2-\dots-x_i^p)\le f(A), \forall i\in \{1,2,\dots,k\}
\]
(Note that in case $f$ is additive the above inequality just becomes: $f(y_i)\le f(x_i^1)+f(x_i^2)+\dots+f(x_i^p)$).
Since $(A-\cup_{j=1}^p x^j_i)\subseteq (A+y_i-\cup_{j=1}^p x^j_i)$, using the submodularity for adding $\cup_{j=1}^p x^j_i$, we get:
\[
f(A+y_i)-f(A+y_i-\cup_{j=1}^p x^j_i)\le f(A)-f(A-\cup_{j=1}^p x^j_i)
\]
and adding $f(A+y_i)-f(A)$ to both sides and using submodularity and the local search stopping condition, we get:
\[
f(A+y_i)-f(A)-f(A)+f(A-\cup_{j=1}^p x^j_i)\le f(A+y_i-\cup_{j=1}^p x^j_i)-f(A)\le 0
\]
We conclude: $f(A+y_i)-f(A)\le f(A)-f(A-\cup_{j=1}^p x^j_i), \forall i\in \{1,2,\dots,k\}$ and adding these inequalities ($\delta(x_i^j)$ is the marginal gain by adding $x_i^j$ at the point of addition):
\[
f_A(B\sm A)\le \sum_{i=1}^k\sum_{j=1}^p\delta(x_i^j) = \sum_{j=1}^p\sum_{i=1}^k\delta(x_i^j)\le \sum_{j=1}^p f(X^j)\le \sum_{j=1}^p f(A\sm B)\le pf(A\sm B)
\]
Now we can $(p+1)$-perturb the marginals for elements of $A\sm B$:
\[
\tilde{f}(B)=f(B)\le f(A\cup B)\le f(A)+f_A(B\sm A)\le f(A)+pf(A\sm B)=\tilde{f}(A)
\]
which contradicts the ($p+1$)-stability of the instance. Once again, for the case of additive $f$: $f_A(B\sm A)=f(B\sm A)$ and thus $p$-stability is enough to guarantee recovery. ($\tilde{f}(B)=f(B)=f(B\sm A)+f(B\cap A)\le pf(A\sm B) +f(B\cap A)\le f(A)+(p-1)f(A\sm B)=\tilde{f}(A)$, where we $p$-perturbed the instance)
\end{proof}


\section{Counterexamples}\label{sec:counter}

Here are two simple counterexamples that prove the tightness of our Greedy recovery results and our Local Search approximation and recovery results:

\begin{itemize}
\item In the submodular case, we proved greedy recovers $(p+1)$-stable $p$-extendible systems. Here is a simple example of a matroid (1-extendible) where Greedy and Local Search fail to recover the optimal solution even though it is 2-stable (also notice that here, Greedy and Local Search give a 2-approximation):  Take $A_1=\{x,\epsilon_1\}, B_1=\{y\}, A_2=\{\epsilon_2\}, B_2=\{x\}$ as in~\cite{filmus2012power}. Assign $w(x)=w(y)=1$ and $w(\epsilon_1)=\epsilon, w(\epsilon_2)=\epsilon$ for some small $\epsilon>0$ (and so $w(A_1)=1+\epsilon$) and consider the partition matroid whose independent sets can only contain one of $\{A_i,B_i\}, i=1,2$. Observe that $\{A_1,A_2\}$ is a local optimum with value $1+2\epsilon$, whereas the global optimum is $\{B_1,B_2\}$ with value 2. Also notice that the same solution is produced by the Greedy algorithm and that the instance can be $(2-\epsilon')$-stable for any small $\epsilon'>0$.

\item Local Search is a $p^2$-approximation for $p$-extendible systems. Look at \hyperref[counter1]{Figure~\ref{counter1}} for a tight counterexample (just for simplicity, we have the $p=2$ case; it generalizes readily).
\end{itemize}

\begin{figure}[h!]
	\centering
	\includegraphics[width=5cm,height=4cm]{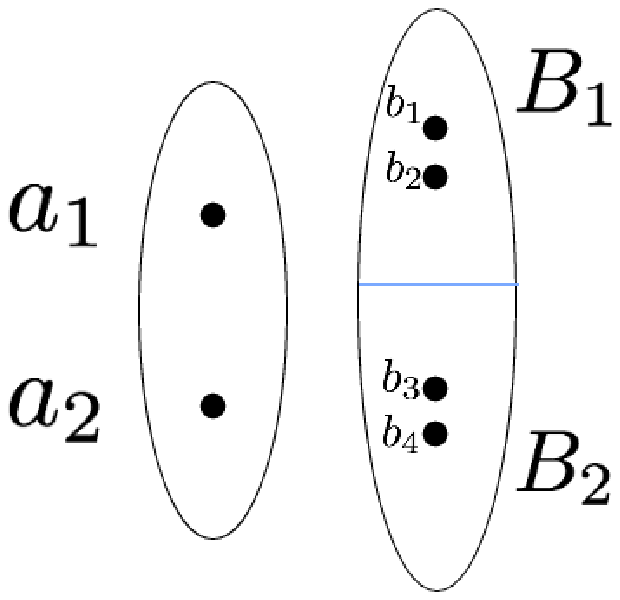}
        	\caption{Local Search is a 4-approximation for this 2-extendible system $(X,\I)$: Let $A=\{a_1,a_2\}$ be feasible and assign $w(a_1)=w(a_2)=1+\epsilon$ and $w(b_i)=2, \forall i\in\{1,2,3,4\}$. The constraints are: $a_1\cup B_1 \notin \I, a_2\cup B_2 \notin \I $, $a_i\cup B_j \in \I$ for $i\neq j$ and $A\cup b_i\notin \I,\forall i\in\{1,2,3,4\}$. Observe that $A$ is a local optimum ($(2,1)$-swaps) with value $2+2\epsilon$, whereas $B_1\cup B_2$ is the global optimum with value 8. Notice also that for the appropriate choice of $\epsilon$, this can be a $(4-\epsilon')$-stable instance for any small $\epsilon'$.}
	\label{counter1}
\end{figure}

\subsection{Cardinality Constraints}\label{sec:cardinality-counter}

Another interesting separation between approximation and stability happens for the case of cardinality constraints. A special case of submodular maximization on $p$-extendible systems is when we have a uniform matroid constraint where the only feasible solutions are those that have cardinality $k \ge 1$ ($\I=\{S\subseteq X: |S|\le k\}$). For this special case, recall that greedy is a $(1-\tfrac{1}{e})$-approximation (in fact, $1-(1-\tfrac{1}{k})^k$) and that this is tight \cite{feige1998threshold}. Regarding stability, we show that the stability threshold needed by greedy for recovery is at least $2-\tfrac{1}{k}$ and so $(1-\tfrac{1}{e})^{-1}$-stability is not enough, i.e. here the approximation threshold is strictly smaller than the stability threshold needed for recovery (see also \hyperref[fig:cardinality]{Figure~\ref{fig:cardinality}}).

\begin{proposition}
For submodular maximization under a uniform matroid ($\I=\{S: |S|\le k\},k\ge 1$), greedy cannot recover $\gamma$-stable instances if $\gamma<(2-\tfrac{1}{k})$.
\end{proposition}

\begin{proof}
The $(2-\tfrac{1}{k}-\delta)$-stable counterexample (for any small $\delta$) where greedy fails is the following: We have in total ($k+1$) elements: $x_1,x_2,\dots, x_k$ and a special element $e$. Denote $O=\{x_1,x_2,\dots, x_k\}$ and with $O_i$ any subset of $O$ with exactly $i$ elements. The function $f$ has: $f(O)=1, f_{O_i}(x_j)=\tfrac{1}{k}, \forall x_j\in O\setminus O_i,  f_{\{e\}\cup O_i}(x_j)=\tfrac{1}{k}(1-\tfrac{1}{k}), \forall x_j\in O\setminus O_i $ and $f(e)=\tfrac{1}{k}, f_{O_i}(e)=\tfrac{1}{k}-\tfrac{i}{k^2}$. Then Greedy first picks element $e$ (to break ties we could set $f(e)=\tfrac{1}{k}+\epsilon$) and then $k-1$ other elements $O_{k-1}\subseteq O$ (let $S=\{e\}\cup O_{k-1}$). However, the optimal solution is $O$ with $f(O)=1$ and greedy has value $1-(\tfrac{1}{k}-\tfrac{1}{k^2})$. Since $S\sm O=\{e\}$, any perturbation such that $\tilde{f}(S)\ge \tilde{f}(O)$ could only $\gamma$-perturb the value $f(e)$: $\tilde{f}(S)\ge \tilde{f}(O)\iff (\gamma-1)\tfrac{1}{k}\ge\tfrac{1}{k}-\tfrac{1}{k^2} \iff \gamma\ge (2-\tfrac{1}{k})$.
\end{proof}
\begin{figure}[h!]
	\centering
	\includegraphics[width=13cm,height=4cm]{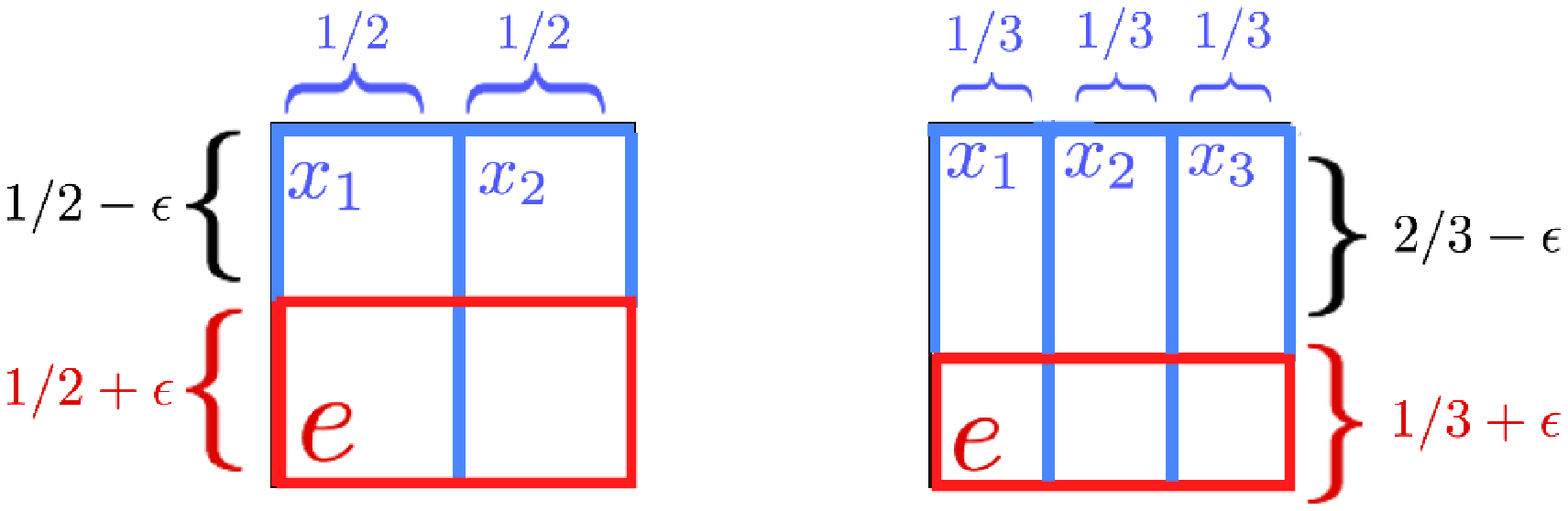}
        	\caption{This is the case for $k=2$ and $k=3$ (the area corresponds to marginal improvements). For $k=2$, there are three elements: $\{e,x_1,x_2\}$. $f(\{x_1,x_2\})=1$, so the optimal solution is $O=\{x_1,x_2\}$. We trick the greedy algorithm which first chooses $\{e\}$ that has slightly better marginal value. For exact recovery, a $\tfrac{3}{2}$-perturbation is needed, even though Greedy is a $\left(\tfrac{4}{3}\right)^{-1}$-approximation. Similarly, for $k=3$, the optimum is $O=\{x_1,x_2,x_3\}$, whereas Greedy picks $\{e,x_1,x_2\}$ and needs $\tfrac{5}{3}$-stability for recovery, even though it is $\left(\tfrac{19}{27}\right)^{-1}$-approximation. Note that stability thresholds need to be larger than the approximation factors.}
	\label{fig:cardinality}
\end{figure}

\end{document}